\documentclass[a4paper,USenglish,autoref,cleveref]{lipics-v2021}

\newif\iflong
\longtrue %

\usepackage[
	vlined,
	linesnumbered,
	noresetcount,
]{algorithm2e}
\DontPrintSemicolon
\SetArgSty{textnormal}
\SetCommentSty{textnormal}
\SetKwBlock{SubAlgoBlock}{}{end}
\newcommand{\SubAlgo}[2]{#1 \SubAlgoBlock{#2}}
\SetAlgorithmName{Algorithm}{Algorithm}{List of algorithms}

\usepackage{cite}

\renewcommand{\_}{\texttt{\textunderscore}}

\newcommand{\pretext}{\textbf{pre:}}
\newcommand{\pre}[1]{\pretext{} #1 \\}

\newcommand{\returnalso}{\ensuremath{\wedge \,}\\\phantom{\textbf{return} }}

\SetKwComment{Comment}{// }{}
\SetKw{Trigger}{trigger}
\SetKw{Decide}{decide}
\SetKw{When}{when}
\SetKw{WhenReceived}{when received}
\SetKw{WhenAccepted}{when accepted}
\SetKw{WhenDelivered}{when delivered}
\SetKw{KwTo}{to}
\SetKw{KwAll}{all}
\SetKw{KwContinue}{continue}
\SetKw{Send}{send}
\SetKw{Broadcast}{broadcast}
\SetKw{KwFrom}{from}
\SetKw{Upon}{upon}
\SetKw{Fun}{function}
\SetKw{WhenTimerFires}{when the timer expires}
\SetKw{WhenTimer}{when timer}
\SetKw{Expires}{expires}
\SetKw{True}{true}
\SetKw{False}{false}
\SetKw{Periodically}{periodically}

\newcommand{\func}[1]{\ensuremath{\mathtt{#1}}}
\newcommand{\var}[1]{\ensuremath{\mathit{#1}}}
\newcommand{\gvar}[1]{\ensuremath{\mathsf{#1}}}

\newcommand{\mess}[1]{\ensuremath{\mathtt{\MakeUppercase{#1}}}}
\newcommand{\lft}{\ensuremath{\leftarrow}\ }

\usepackage{float}
\usepackage{placeins}

\usepackage{multicol}

\usepackage[dvipsnames]{xcolor}

\theoremstyle{plain}

\newcommand{\GST}{\mathsf{GST}}
\newcommand{\valid}{\mathsf{valid}}
\newcommand{\leader}{\mathsf{leader}}

\let\oldnl\nl%
\newcommand{\nonl}{\renewcommand{\nl}{\let\nl\oldnl}}

\newcommand{\tr}[2]{\iflong{}\S#1\else{}\cite[\S#2]{ext}\fi}
\newcommand{\tra}[2]{\iflong{}(\S#1)\else{}\cite[\S#2]{ext}\fi}

\nolinenumbers

\title{Fast TetraBFT: Optimizing Latency Where It Matters}

\author{Antonio J. Fernández-Pinto}{IMDEA Software Institute, Madrid, Spain \and
  Universidad Politécnica de Madrid, Spain}{}{https://orcid.org/0009-0003-6712-3996}{}
\author{Manuel Bravo}{Informal Systems, Madrid, Spain}{}{https://orcid.org/0000-0002-8194-8789}{}
\author{Gregory Chockler}{University of Surrey, Guildford, UK}{}{https://orcid.org/0000-0001-6700-9235}{}
\author{Alexey Gotsman}{IMDEA Software Institute, Madrid, Spain}{}{https://orcid.org/0000-0002-2454-4171}{}

\titlerunning{Fast TetraBFT: Optimizing Latency Where It Matters}
\authorrunning{A.\,J. Fernández-Pinto, M. Bravo, G. Chockler, and A. Gotsman}
\Copyright{Antonio J. Fernández-Pinto, Manuel Bravo, Gregory Chockler, and Alexey Gotsman}



\funding{This work was partially supported by the projects SACO and María de
  Maeztu, funded by MICIU/AEI, and REDONDA, funded by the EU CHIST-ERA network.}

\EventEditors{Keren Censor-Hillel, and William K. Moses Jr.}
\EventNoEds{4}
\EventLongTitle{40th International Symposium on Distributed Computing (DISC 2026)}
\EventShortTitle{DISC 2026}
\EventAcronym{DISC}
\EventYear{2026}
\EventDate{November 9--13, 2026}
\EventLocation{Rome, Italy}
\EventLogo{}
\SeriesVolume{397}
\ArticleNo{56}

\begin{document}

\maketitle

\begin{abstract}
  Unauthenticated Byzantine consensus protocols achieve optimal failure
  resilience while relying only on authenticated point-to-point channels, not
  authenticated messages. They are an attractive building block for blockchains
  without global trust assumptions as well as for future post-quantum
  settings. We consider unauthenticated Byzantine consensus in the partial
  synchrony model and focus on optimizing its \emph{good-case latency} -- the
  worst-case time for correct processes to reach a decision under favorable
  conditions. A recently proposed ForgetIT protocol achieves an optimal
  good-case latency of $3$ message delays but employs a highly complex
  design. We show that this complexity is unnecessary. To this end, we present
  Fast TetraBFT -- an unauthenticated Byzantine consensus protocol that achieves
  optimal good-case latency by augmenting an existing TetraBFT protocol with a
  simple fast-path wrapper. Our solution lowers the good-case latency of
  TetraBFT from $5$ to $3$ message delays while preserving its bounded space
  requirements and low communication complexity.
\end{abstract}

\bigskip

\section{Introduction}

We consider Byzantine consensus in the \emph{partial synchrony} model~\cite{dls}:
the system is guaranteed to eventually reach a \emph{Global Stabilization Time}
($\GST$) after which message delays are bounded by an unknown constant $\delta$
and process clocks track real time.
Motivated by the blockchain use case, in recent years researchers have developed
a plethora of Byzantine consensus protocols in this model. Most of these are
\emph{authenticated}, i.e., they send messages authenticated using cryptographic
signatures. These allow a process $p_1$ to prove to a process $p_2$ that a third
process $p_3$ has sent a given message.
In contrast, \emph{unauthenticated} (aka \emph{information-theoretic}) protocols
rely only on authenticated point-to-point channels and thus disallow the above
pattern.
The minimal assumptions of such protocols make them applicable in a broader
range of settings.  These include post-quantum %
or computationally constrained environments, as well as blockchains without
global trust assumptions, such as Stellar~\cite{stellar-protocol} and
Ripple~\cite{ripple}. %

Since realistic networks behave synchronously most of the time, the efficiency of
partially synchronous blockchain implementations is largely determined by their performance under stable
network conditions.
In more detail, blockchains often run a sequence of Byzantine consensus
instances, each deciding on the next block in the chain.
A consensus instance proceeds through a sequence of \emph{views} $0, 1, 2, \ldots$, each with a designated leader process responsible for driving the protocol towards a decision.  Since in practice the leaders are selected to maximize the chances of correct behavior (e.g., by penalizing misbehavior via stake slashing), under stable network conditions decisions are usually reached in view $0$.
For example, data from the Cosmos Hub blockchain~\cite{cosmos} running the
Tendermint protocol~\cite{tendermint} shows that 99.995\% of the instances
were decided in view $0$ over a seven-day sample period. Thus, a
performance-critical metric for Byzantine consensus protocols is their \emph{good-case latency} -- 
the worst-case latency to decision if the consensus
instance is started during the synchronous period (after $\GST$) and the
initial leader is correct. 

In this paper we consider the problem of optimizing the good-case latency of
unauthenticated Byzantine consensus.  Assuming the optimal resilience --
$n \ge 3f+1$ processes to tolerate $f$ Byzantine failures -- the lower bound on
the good-case latency is $3\delta$~\cite{unauthenticated-good-case}.
The first protocol to match this bound, a signature-free version of
PBFT~\cite{castro-thesis}, was able to achieve this only at the expense of a
high message complexity -- $O(n^3)$ per view.  Information-theoretic
HotStuff~\cite{it-hotstuff} and TetraBFT~\cite{tetrabft} then lowered the
communication complexity to $O(n^2)$ per view, but only at the expense of a
higher good-case latency -- $6\delta$ and $5\delta$, respectively. A very recent
proposal, Forget-IT~\cite{forgetit}, achieved both the optimal good-case latency
of $3\delta$ and the communication complexity of $O(n^2)$ per
view. Unfortunately, the protocol has a complicated design, in part because it
also minimizes another latency metric, which its authors call \emph{robust
  good-case latency}. This is the worst-case latency of a decision measured from
the first time after $\GST$ when a correct leader enters a view (whether $0$ or
higher).

The rationale for the above metric was to capture latency guarantees under 
adversarial conditions, when the protocol is started before $\GST$.
We argue that the metric does not fulfil this function.
In practice, what matters under adversarial conditions is how quickly the protocol recovers after the onset of the synchronous period, i.e., the time to decision measured from $\GST$.
To take a decision in this case, processes need to first exit the views they were in before $\GST$ and enter a ``good'' view with a correct leader.
Processes exit a view when their timers expire, and the duration of these timers is typically orders of magnitude larger than $\delta$.
Thus, optimizing the latency to the decision after entering a good view is irrelevant, since the time to decision after $\GST$ will be dominated by timer durations.

We set out to find a simpler protocol that achieves the optimal good-case latency and low communication complexity.
Our answer is \emph{Fast TetraBFT} -- a protocol with the good-case latency of $3\delta$ and the communication complexity of $O(n^2)$ per view.
The protocol guarantees optimal resilience and optimistic responsiveness~\cite{optimistic}, and can be implemented in bounded space \tra{\ref{app:bounded}}{A}. %
It is a wrapper around TetraBFT~\cite{tetrabft} that uses a special code path for view $0$ and executes a slightly modified TetraBFT in higher views.
Thus, our proposal improves the good-case latency of TetraBFT while preserving its performance under adversarial conditions.

\section{System Model}

We consider a distributed message-passing system %
of $n \ge 3f+1$ processes, of which at most $f$ can be Byzantine.
The system is partially synchronous~\cite{dls,ChandraToueg1996}: for every
execution, there exist a \emph{Global Stabilization Time} ($\GST$) and a
duration $\delta$ such that, after $\GST$, message delays between every pair of
correct processes are bounded by $\delta$; before $\GST$, messages can be
arbitrarily delayed or lost. We also assume a known upper bound $\Delta$ on the
maximum value of $\delta$ across all executions, i.e.,
$\delta < \Delta$~\cite{HerzbergKutten1989}. %
Processes are equipped with hardware clocks that can drift arbitrarily from real
time before $\GST$, but do not drift thereafter. Communication channels are
\emph{authenticated}: %
if a correct process $p_j$ receives a message $m$ claiming to have been sent by
a process $p_i$, and $p_i$ is correct, then $m$ was previously sent by $p_i$ to
$p_j$.

We assume a predicate $\valid()$ over consensus proposals. In Byzantine
consensus, each correct process proposes a valid value and then has to decide on
a value so that: (\textsf{Validity}) a correct process decides on a valid value;
(\textsf{Agreement}) no two correct processes decide on different values; and
(\textsf{Liveness}) every correct process eventually decides on a value.

\begin{figure}[t]
	\begin{minipage}{0.58\linewidth}
	\scalebox{0.85}{
	\begin{algorithm}[H]
		{\color{red} \tcp{The following handlers are executed only if $\gvar{started\_tetra} = \True$}}
		\SubAlgo{\Fun \func{init\_tetra}($x$)}{
			$\gvar{val}$ \lft $x$\\
			Start the view synchronizer\\
		}
		\SubAlgo{\Upon \func{new\_view}($v$) \label{line:tetra-new-view}}{
			$\gvar{view}$ \lft $v$ \\
			$\gvar{voted}$ \lft \False \\
			\Send \mess{suggest}($\gvar{view}$, $\gvar{V2}$,
                        $\gvar{prev\_V2}$, $\gvar{V3}$) \qquad \KwTo $\leader(\gvar{view})$ \label{line:tetra-suggest}\\
			\Send \mess{proof}($\gvar{view}$, $\gvar{V1}$, $\gvar{prev\_V1}$, $\gvar{V4}$) \KwTo \KwAll\label{line:tetra-proof}\\
		}
		{\color{red}
		\SubAlgo{\WhenReceived $f + 1$ $\mess{vote2}(\_, x_j)$, \qquad directly or within \mess{suggest} messages}{ \label{line:witness}
			\pre{$\gvar{lock} \neq \bot \wedge \forall j.\, \gvar{lock} \neq x_j$} \label{line:pre-unlock}
			$\gvar{lock}$ \lft $\bot$\\ \label{line:unlock2}
		}
		}
		\SubAlgo{\WhenReceived $n-f$ $\mess{suggest}(\gvar{view}, \_, \_,\_)$}{
                  \pre{$p_i = \leader(\gvar{view}) \wedge \exists x.\, \valid(x)\wedge {}$\label{line:propose-pre}}
                  \nonl\phantom{{\bf pre: }}$\func{safe\_val\_leader}(x) \wedge (x = \gvar{val} \vee {}$\\
                  \nonl\phantom{{\bf pre: }}$x$ is in one of the \mess{suggest} messages$)$\\
			\Send \mess{propose}($\gvar{view}$, $x$) \KwTo \KwAll \label{line:tetra-propose}\\
                      }

                      \smallskip
		\end{algorithm}
		}
		\end{minipage}
		\hspace{-30pt}
		\begin{minipage}{0.6\linewidth}
		\scalebox{0.85}{
		\begin{algorithm}[H]
		\SubAlgo{\WhenReceived $\mess{propose}(\gvar{view}, x)$
                  \KwFrom $\leader(\gvar{view})$ {\bf and} $n-f$
                  $\mess{PROOF}(\gvar{view}, \_, \_, \_)$}{
			\pre{$\func{safe\_val\_follower}(x) \wedge {}$} \label{line:pre-vote1}
                        \nonl\phantom{{\bf pre: }}$\valid(x) \wedge \neg \gvar{voted}$\\
			$\gvar{voted}$ \lft \True \\
			\lIf{the value in $\gvar{V1}$ $\not= x$\label{line:different-vote1}}{$\gvar{prev\_V1}$ \lft $\gvar{V1}$}
			$\gvar{V1}$ \lft \mess{vote1}($\gvar{view}$, $x$) \\
			\Send $\gvar{V1}$ \KwTo \KwAll \label{line:tetra-vote1}\\
		}
		\SubAlgo{\WhenReceived $n - f$ \mess{vote1}($\gvar{view}$,
                  $x$) \label{line:receive-vote1}}{\lIf{the value in $\gvar{V2}$
                    $\not= x$\label{line:different-vote2}}{$\gvar{prev\_V2}$ \lft $\gvar{V2}$}
			$\gvar{V2}$ \lft \mess{vote2}($\gvar{view}$, $x$) \\
			\Send $\gvar{V2}$ \KwTo \KwAll \\
		}
		\SubAlgo{\WhenReceived $n - f$ \mess{vote2}($\gvar{view}$, $x$)}{
			$\gvar{V3}$ \lft \mess{vote3}($\gvar{view}$, $x$) \\
			\Send $\gvar{V3}$ \KwTo \KwAll \\
		}
		\SubAlgo{\WhenReceived $n - f$ \mess{vote3}($\gvar{view}$, $x$)}{
			$\gvar{V4}$ \lft \mess{vote4}($\gvar{view}$, $x$) \\
			\Send $\gvar{V4}$ \KwTo \KwAll \label{line:send-vote4}\\
		}
		\SubAlgo{\WhenReceived $n - f$ \mess{vote4}($\gvar{view}$, $x$)}{
			{\bf decide $x$} \label{line:tetra-decide} \\
		}
	\end{algorithm}
	}
	\end{minipage}

        \caption{TetraBFT at process $p_i$. Initially, $\gvar{V1}, \gvar{V2}, \gvar{V3}, \gvar{V4}$, $\gvar{prev\_V1}$, $\gvar{prev\_V2}$ are $\bot$.}
\label{alg:tetra}
\end{figure}

\begin{figure}[t]
	\begin{minipage}{1.3\linewidth}
	\scalebox{0.85}{
	\begin{algorithm}[H]
		\SubAlgo{\Fun \func{safe\_suggest}($x$, $v$, $\mess{suggest}(\_, \mess{vote2}(v_2, x_2), \mess{vote2}(v_2', x_2'), \_)$)}{
			\Return $v = 1 \vee (v_2 \geq v \wedge x_2 = x) \vee (v_2' \geq v \wedge x_2' \neq x_2)$ \label{line:safe_suggest} \\
		}
		\SubAlgo{\Fun \func{safe\_proof}($x$, $v$, $\mess{proof}(\_, \mess{vote1}(v_1, x_1), \mess{vote1}(v_1', x_1'), \_)$)}{
			\Return $v = 1 \vee (v_1 \geq v \wedge x_1 = x) \vee (v_1' \geq v \wedge x_1' \neq x_1)$ \label{line:safe_proof} \\
		}
		\SubAlgo{\Fun \func{blocking\_safe\_suggests}($x$, $v$)}{ \label{line:block_suggests}
			\Return $\exists B.\, |B| \geq n \,{-}\, 2f \wedge \forall
	                p \,{\in}\, B.\,$\mbox{received a \mess{suggest} message $m$ from $p$
	                  s.t.} $\func{safe\_suggest}(x, v, m)$\\
		}
		\SubAlgo{\Fun \func{blocking\_safe\_proofs}($x$, $v$)}{
			\Return $\exists B.\, |B| \geq n \,{-}\, 2f \wedge \forall p \,{\in}\, B.\,$\mbox{received a \mess{proof} message $m$ from $p$
	                  s.t.} $\func{safe\_proof}(x, v, m)$ \\
		}
		\SubAlgo{\Fun \func{safe\_val\_leader}($x$)}{

			{\color{red}\lIf{$\gvar{lock} \neq \bot \wedge \gvar{lock} \neq x$}{\Return \False}\label{line:safe_leader_lock_check}}
			\lIf{$\gvar{view} = 1$}{\Return \True}
			Examine the information about \mess{vote3} received in this \gvar{view}'s \mess{suggest} messages:\\
			\lIf{$n - f$ processes have not sent any \mess{vote3} before $\gvar{view}$}{\Return \True} \label{line:no_vote3}
			\Return $\exists Q, v'.\, |Q| \geq n - f \wedge v' < \gvar{view}
	                  \wedge \func{blocking\_safe\_suggests}(x, v')$\label{tetra:safe-val-leader-3}
				\nonl \returnalso (no process in $Q$ has sent a \mess{vote3}
	                        after $v'$ or a \mess{vote3}$(v', x')$ for any $x' \neq x$) \\
		}
		\SubAlgo{\Fun \func{safe\_val\_follower}($x$)}{
			{\color{red}\lIf{$\gvar{lock} \neq \bot \wedge \gvar{lock} \neq x$}{\Return \False}\label{line:safe_follower_lock_check}}
			\lIf{$\gvar{view} = 1$}{\Return \True}
			Examine the information about \mess{vote4} received in this \gvar{view}'s \mess{proof} messages:\\
			\lIf{$n - f$ processes have not sent any \mess{vote4} before $\gvar{view}$\label{tetra:safe-val-follower-2}}{\Return \True}
			\Return $\exists Q, v'.\, |Q| \geq n - f \wedge v' <
                        \gvar{view}$ \label{tetra:safe-val-follower-3}
				\nonl \returnalso (no process in $Q$ has sent a \mess{vote4}
	                        after $v'$ or a \mess{vote4}$(v', x')$ for any $x' \neq x$)
	                  \nonl \returnalso $(\func{blocking\_safe\_proofs}(x, v') \vee 
	(\exists x_1, x_2, v_1, v_2.\, x_1 \neq x_2 \wedge v' \leq v_1 < v_2 < \gvar{view}$
			\nonl \returnalso $\func{blocking\_safe\_proofs}(x_1, v_1) \wedge \func{blocking\_safe\_proofs}(x_2, v_2)))$ \\
		}
	\end{algorithm}
      }

      \vspace{-6pt}

\end{minipage}
\caption{TetraBFT: auxiliary functions.}
\label{alg:tetra_aux}
\end{figure}

\section{TetraBFT Overview}
\label{sec:overview}

We build on the TetraBFT protocol~\cite{tetrabft}, whose pseudocode is given in Figures~\ref{alg:tetra}--\ref{alg:tetra_aux} (the code highlighted in red is necessary for Fast TetraBFT and should be ignored for now).
The protocol proceeds in views, and each view $v$ has a designated leader $\leader(v)$ whose role rotates round robin among all processes.
We number TetraBFT views starting from $1$, since we use view $0$ for Fast TetraBFT.
Processes enter a view $v$ upon an upcall $\func{new\_view}(v)$ from a \emph{view synchronizer} (line~\ref{line:tetra-new-view}), whose implementation we omit.
For concreteness, we could use the FastSync synchronizer of Bravo et al.~\cite{bftlive-dc}, which does not rely on authenticated messages and requires only bounded space. 
The synchronizer ensures that for some view $\mathcal{V}$ entered after $\GST$, all correct processes enter each view $v \geq \mathcal{V}$ and overlap in $v$ for a given duration (TetraBFT requires $6\delta$).
The structure of each view in TetraBFT is as follows:
\begin{itemize}
\item Upon entering a new view, processes send a \mess{suggest} message to its leader with information to help it make a proposal (line~\ref{line:tetra-suggest}).
  They also broadcast \mess{proof} messages to help other processes check that the leader's proposal is correct (line~\ref{line:tetra-proof}).
\item The leader determines its proposal based on a \emph{quorum} of $n-f$ \mess{suggest} messages satisfying \func{safe\_val\_leader}, and broadcasts the proposal in a \mess{propose} message (line~\ref{line:tetra-propose}).
\item
A process accepts the leader's proposal if it receives a quorum of \mess{proof} messages satisfying \func{safe\_val\_follower}. A process accepts only one proposal per view (cf. the \gvar{voted} flag).
\item
Processes go through four rounds of voting on the proposal, sending messages \mess{vote1} to \mess{vote4} (line~\ref{line:tetra-vote1}--\ref{line:send-vote4}).
Votes in each succeeding round are sent in response to a quorum of matching votes from the previous round.
When a process receives a quorum of matching \mess{vote4} messages, it decides (line~\ref{line:tetra-decide}).
\item
A process stores the last votes it issued in variables \gvar{V1} to \gvar{V4}.
For \mess{vote1} and \mess{vote2} it also keeps track of the last two votes for distinct values using
\gvar{prev\_V1} and \gvar{prev\_V2} (lines~\ref{line:different-vote1} and~\ref{line:different-vote2}).
These votes are then reported in \mess{suggest} and \mess{proof} messages (lines~\ref{line:tetra-suggest}--\ref{line:tetra-proof}).
\end{itemize}

\section{Fast TetraBFT}

Fast TetraBFT runs a special protocol in view $0$ (\autoref{alg:view_0}), which is driven by an \emph{initial leader}.
When this leader is correct and the network is synchronous, the protocol decides within $3\delta$.
Hence, at the beginning of the protocol each process sets a timer for a conservative bound of $3\Delta$ (line~\ref{line:set-timer}).
When the timer expires (line~\ref{line:timeout}), the process switches to TetraBFT, disabling the actions in \autoref{alg:view_0} and enabling those in \autoref{alg:tetra} via a flag \gvar{started\_tetra}.

The initial leader starts Fast TetraBFT by broadcasting its consensus proposal in a \mess{fast\_propose} message (line~\ref{line:send-fast-propose}).
Each process accepts only one such proposal (cf. the \gvar{voted} flag), which it broadcasts in a \mess{vote0} message (line~\ref{line:vote}).
A process receiving a quorum of \mess{vote0}s for the same value broadcasts it in a \mess{commit} message (line~\ref{line:vote-end}), and a process receiving a quorum of matching \mess{commit}s decides (line~\ref{line:early-decide}).
This ensures the good-case latency of $3\delta$.

A process sending a \mess{commit} message for a value $x$ agrees to $x$ being decided, and thus becomes \emph{locked} on this value.
This is recorded by storing $x$ in \gvar{val} and \gvar{lock} (line~\ref{line:lock}).
The former variable determines the proposal fed into TetraBFT (line~\ref{line:init-tetra}), and is initialized to the process's own proposal (line~\ref{line:val-init}).
The latter constrains the proposals that can be made by leaders and accepted by followers in TetraBFT: when \gvar{lock} is set to $x$, any proposals different from $x$ are rejected (lines~\ref{line:safe_leader_lock_check} and~\ref{line:safe_follower_lock_check}).
Since view $0$ may fail to reach a decision, a process also needs to be able to \emph{unlock} the value $x$ it recorded, to allow TetraBFT to decide.
Our main insight is that the process can do this upon receiving $f + 1$ \mess{vote2} messages for values different from $x$, either directly or within \mess{suggest} messages (line~\ref{line:witness}).

\begin{figure}[t]
\begin{minipage}{\linewidth}
	\begin{minipage}{0.5\linewidth}
	\scalebox{0.85}{
	\begin{algorithm}[H]
		\SubAlgo{\Fun \func{consensus}($x$)}{
			$\gvar{val}$ \lft $x$ \label{line:val-init} \\
			$\gvar{lock}$ \lft $\bot$ \\
			$\gvar{started\_tetra}$ \lft \gvar{voted} \lft \False \\
			Set a timer for $3\Delta$ \label{line:set-timer}\\
			\If{$p_i$ is the initial leader}{
				\Send \mess{fast\_propose}($\var{x}$) \KwTo \KwAll \label{line:send-fast-propose}\\
			}
		}
		\SubAlgo{{\bf when} the timer expires}{ \label{line:timeout}
			\pre{$\neg \gvar{started\_tetra}$}
			$\gvar{started\_tetra}$ \lft \True \\
			\func{init\_tetra}($\gvar{val}$) \label{line:init-tetra}\\
		}
	\end{algorithm}
	}
	\end{minipage}
	\hfill
	\begin{minipage}{0.6\linewidth}
	\scalebox{0.85}{
	\begin{algorithm}[H]
		\SubAlgo{\WhenReceived \mess{fast\_propose}($x$) \KwFrom \qquad the initial leader}{
			\pre{$\valid(x) \wedge \neg \gvar{voted} \wedge \neg \gvar{started\_tetra}$} \label{line:pre-vote} 
			$\gvar{voted}$ \lft \True \\
			\Send \mess{vote0}($x$) \KwTo \KwAll \\ \label{line:vote}
		}
		\SubAlgo{\WhenReceived $n - f$ \mess{vote0}($x$)\label{line:receive-vote}}{
			\pre{$\neg \gvar{started\_tetra}$} \label{line:pre-lock}
			$\gvar{lock}$ \lft $\gvar{val}$ \lft $x$ \label{line:lock}\\
			\Send \mess{commit}($x$) \KwTo \KwAll \label{line:vote-end}\\
		}
		\SubAlgo{\WhenReceived $n - f$ \mess{commit}($x$) \label{line:receive-final-vote}}{
			{\bf decide $x$} \label{line:early-decide}\\
                      }
	\end{algorithm}
	}
	\end{minipage}

        \vspace{-3pt}
\end{minipage}

\caption{Fast TetraBFT code for view $0$ at process $p_i$.}
\label{alg:view_0}
\end{figure}

\section{Correctness}

\textsf{Validity} holds trivially, so we now prove \textsf{Agreement}. The
safety of TetraBFT rules out contradictory decisions in views $>0$
\tra{\ref{app:proofs}}{B}. Lemmas~\ref{lem:locks-agree}-\ref{lem:view0-agree} below
prove agreement between values decided in view $0$. Finally,
Lemmas~\ref{lem:no-unlock}-\ref{lem:agree-tetra-wrapper} prove agreement between
the values decided in view $0$ (line~\ref{line:early-decide}) and those decided
in the subsequent views (line~\ref{line:tetra-decide}).

\begin{lemma}
  If two correct processes lock values $x$ and $x'$, respectively, then $x=x'$.
\label{lem:locks-agree}
\end{lemma}
\begin{proof}
A correct process locks a
value only after receiving a quorum of $n-f$ $\mess{VOTE0}$ messages for it
(line~\ref{line:receive-vote}). Since any two quorums intersect in a correct
process, locking different values $x$ and $x'$ requires some correct process to
send both $\mess{VOTE0}(x)$ and $\mess{VOTE0}(x')$. This cannot happen since a
process only sends a $\mess{VOTE0}$ once (the \gvar{voted} flag at
line~\ref{line:pre-vote}).
\end{proof}

\begin{lemma}
Any two correct processes that decide in view $0$ must decide the same value.
\label{lem:view0-agree}
\end{lemma} 
\begin{proof}
A correct process decides a value $x$ in view $0$ after receiving
$\mess{commit}(x)$ from a quorum of $n-f$ processes
(line~\ref{line:receive-final-vote}). Out of these processes, at least one must
be correct, and this process locks $x$ before sending its message.
The required then follows from Lemma~\ref{lem:locks-agree}.
\end{proof}

\begin{lemma}
  If $x$ is locked by $f + 1$ correct processes, then no correct process can
  unlock $x$.
  \label{lem:no-unlock}
\end{lemma}
\begin{proof}
  Assume toward a contradiction that some correct process unlocks $x$. Let $t$
  be the earliest time when this happens, and let $p_i$ be the unlocking
  process.  Then by the time $t$, $p_i$ has received $\mess{VOTE2}(\_, x_j)$
  from $f+1$ processes $p_j$ such that every $x_j\neq x$
  (line~\ref{line:witness}). Since at least one of these processes is correct,
  it must have sent its $\mess{VOTE2}$ following the receipt of $n-f$ matching
  $\mess{VOTE1}$s for some value $x'\neq x$ (line~\ref{line:receive-vote1}). We
  know that $f + 1$ correct processes lock $x$, so there exists a correct process
  $p$ which both locks $x$ at some time $t_l$ and broadcasts
  $\mess{VOTE1}(\_, x')$ at some time $t' < t$. Then
  $\func{safe\_val\_follower}(x', v')$ holds at $p$ at time $t'$. Since
  $x$ is the only value that $p$ can ever lock (Lemma~\ref{lem:locks-agree}), 
  line~\ref{line:safe_follower_lock_check} implies that $p$ has
  $\gvar{lock}=\bot$ at time $t'$.
  Furthermore, $\gvar{started\_tetra}=\False$ at $p$ at $t_l$
  (line~\ref{line:pre-lock}) and must be $\True$ for $p$ to send $\mess{VOTE1}$
  at $t'$.
  Then $t_l < t'$, so $p$ must unlock $x$ at some point between $t_l$
  and $t'$. Since $t' < t$, this contradicts our assumption that no correct
  process does so before $t$.
\end{proof}

\begin{lemma}
  If two correct processes decide $x$ in view $0$ and $x'$ in a view $>0$, then
  $x=x'$.
\label{lem:agree-tetra-wrapper}
\end{lemma}
\begin{proof}
  A correct process decides $x$ in view $0$ after receiving a quorum of $n-f$
  $\mess{COMMIT}(x)$ messages (line~\ref{line:receive-final-vote}).  Out of
  these, at least $n-2f$ must be sent by correct processes, which have thus
  locked $x$ (line~\ref{line:lock}). By Lemma~\ref{lem:no-unlock}, these
  processes will never unlock $x$. Then \func{safe\_val\_follower} ensures that
  they will never send a $\mess{VOTE1}$ message for another value $x'\neq x$ in
  any view $>0$ (lines~\ref{line:safe_follower_lock_check}
  and~\ref{line:pre-vote1}). This makes it impossible to assemble a quorum of
  $n-f$ $\mess{VOTE1}$ messages for $x'$ needed to decide this value in a view
  $>0$.
\end{proof}

\begin{theorem}[Liveness] \label{thm:liveness}
	Every correct process eventually decides.
\end{theorem}
\begin{proof}
Each correct process eventually times out and starts executing TetraBFT (line~\ref{line:init-tetra}).
If all correct processes eventually have $\gvar{lock} = \bot$, then they end up executing unmodified TetraBFT, which we know is live (we formalize this in \tr{\ref{app:proofs}}{B}).
Thus, we only have to consider the case when a non-empty set of correct processes $L$ get locked on a value and never unlock.
By Lemma~\ref{lem:locks-agree}, they must all be locked on the same value $l$.
By the properties of the view synchronizer (\S\ref{sec:overview}), there exists a view $v \ge \mathcal{V}$ such that: the leader of $v$ is a correct process $p \in L$; and all correct processes enter $v$ after $\GST$ and overlap in $v$ for at least $6\delta$.
We now show that: \emph{(i)} $p$ must propose $l$; and \emph{(ii)} all correct processes accept this proposal.
Then the interval of $6\delta$ during which the correct processes overlap in $v$ is enough for them to exchange the messages leading to a decision, as required.

To prove \emph{(i)}, note that, by \autoref{line:safe_leader_lock_check}, $p$ cannot propose any value other than $l$.
We now prove that $p$ will send $\mess{PROPOSE}(v, l)$.
Assume this is not the case.
Then by \autoref{line:propose-pre}, $\func{safe\_val\_leader}(l)$ returns \False.
Process $p$ will eventually receive \mess{suggest} messages for view $v$ from a quorum $Q$ of $n - f$ correct processes.
By \autoref{line:no_vote3}, one of the processes in $Q$ must report in its \mess{suggest} message that it has previously sent a \mess{vote3}: otherwise $\func{safe\_val\_leader}(l)$ would return \True.
Let $v' < v$ be the highest view where a process from $Q$ sent \mess{vote3}, and $x$ be the value in this \mess{vote3}.
We now consider two cases:
\begin{itemize}
\item $x \neq l$.
  Correct processes include the two latest \mess{vote2}s for distinct values into their \mess{suggest} messages (\autoref{line:tetra-suggest}).
  Since a correct process sent $\mess{vote3}(v',x)$, $n - 2f$ correct processes sent $\mess{vote2}(v', x)$.
  Then the \mess{suggest} messages these processes send to $p$ in view $v$ must include $\mess{vote2}$ for a value different from $l$.
  This causes $p$ to unlock (\autoref{line:unlock2}), contradicting the fact that $p \in L$.

\item $x = l$.
Since a correct process sent $\mess{vote3}(v',l)$, $n - 2f$ correct processes sent $\mess{vote2}(v', l)$.
Consider one of these processes, and let $m$ be the \mess{suggest} message it sends in view $v$.
If between views $v'$ and $v$, it did not send \mess{vote2} for a value $x \neq l$,
then $\func{safe\_suggest}(l, v', m)$ holds by the second disjunct in \autoref{line:safe_suggest}.
Otherwise, the predicate holds by the third disjunct.
Thus, $\func{safe\_suggest}(l, v', m)$ holds for $n - 2f$ \mess{suggest} messages $m$ sent by correct processes in $v$.
Once $p$ receives these, $\func{blocking\_safe\_suggest}(l, v')$ holds.
This validates the condition at \autoref{tetra:safe-val-leader-3}, so that $\func{safe\_val\_leader}(l)$ returns \True: a contradiction.
\end{itemize}

 The proof of \emph{(ii)} is analogous.
  Assume toward a contradiction that some correct process $q$ does not accept $p$'s proposal.
  Then by \autoref{line:pre-vote1}, $\func{safe\_val\_follower}(l)$ returns \False.
  By Lemma~\ref{lem:locks-agree}, this cannot be due to \autoref{line:safe_follower_lock_check}.
Process $q$ will eventually receive \mess{proof} messages from a quorum $Q$ of correct processes.
By \autoref{tetra:safe-val-follower-2}, one of the processes in $Q$ must report in its \mess{proof} message that it had previously sent a \mess{vote4}.
Let $v' < v$ be the highest view where a process from $Q$ sent \mess{vote4}, and $x$ be the value in this \mess{vote4}.
We now consider two cases:
\begin{itemize}
\item $x \neq l$.
  Correct processes include the two latest \mess{vote2}s for distinct values into their \mess{suggest} messages (\autoref{line:tetra-suggest}).
  Since a correct process sent $\mess{vote4}(v',x)$, $n - 2f$ correct processes sent $\mess{vote2}(v', x)$.
  Then the \mess{suggest} messages these processes send to $p$ in view $v$ must include $\mess{vote2}$ for a value different from $l$.
  This causes $p$ to unlock (\autoref{line:unlock2}), contradicting the fact that $p \in L$.
\item $x = l$.
  Since a correct process sent $\mess{vote4}(v',l)$, $n - 2f$ correct processes sent $\mess{vote1}(v', l)$.
Consider one of these processes, and let $m$ be the \mess{proof} message it sends in view $v$.
If between views $v'$ and $v$, it did not send \mess{vote1} for a value $x \neq l$, then $\func{safe\_proof}(l, v', m)$ holds by the second disjunct in line~\ref{line:safe_proof}.
Otherwise, the predicate holds by the third disjunct.
Thus, $\func{safe\_proof}(l, v', m)$ holds for $n-2f$ \mess{proof} messages $m$ sent by correct processes in $v$.
Once $q$ receives these, $\func{blocking\_safe\_proofs}(l, v')$ holds.
This validates the condition at \autoref{tetra:safe-val-follower-3}, so that $\func{safe\_val\_leader}(l)$ returns \True: a contradiction.
\qedhere
\end{itemize}
\end{proof}

\bibliographystyle{plainurl}
\bibliography{biblio}

\clearpage

\iflong
\appendix

\section{Ensuring Bounded Memory}
\label{app:bounded}

Fast TetraBFT can be implemented in bounded memory as follows:
\begin{itemize}
\item For each message type and sender, a process stores the message of this type
  from this sender that has the highest view.
\item To support unlocking (line~\ref{line:witness}), a process also retains two
  latest \mess{vote2} messages from every sender that were cast for distinct
  values. This applies to \mess{vote2}s sent both directly and within
  \mess{suggest} messages.
\end{itemize}
These rules do not compromise liveness because the protocol does not rely on
messages from previous views for any purposes other than unlocking. The above
rule for \mess{vote2} messages ensures that a process never forgets having
received a \mess{vote2} for a value different from its lock, so the unlocking
condition in lines~\ref{line:witness}-\ref{line:pre-unlock} is never missed.

\section{Additional Proofs}
\label{app:proofs}

\begin{lemma}
  If two correct processes lock values $x$ and $x'$, respectively, then $x=x'$.
\label{lem:locks-agree}
\end{lemma}
\begin{proof}
A correct process locks a
value only after receiving a quorum of $n-f$ $\mess{VOTE0}$ messages for it
(line~\ref{line:receive-vote}). Since any two quorums intersect in a correct
process, locking different values $x$ and $x'$ requires some correct process to
send both $\mess{VOTE0}(x)$ and $\mess{VOTE0}(x')$. This cannot happen since a
process only sends a $\mess{VOTE0}$ once (the \gvar{voted} flag at
line~\ref{line:pre-vote}).
\end{proof}

\begin{lemma}
Any two correct processes that decide in view $0$ must decide the same value.
\label{lem:view0-agree}
\end{lemma} 
\begin{proof}
A correct process decides a value $x$ in view $0$ after receiving
$\mess{commit}(x)$ from a quorum of $n-f$ processes
(line~\ref{line:receive-final-vote}). Out of these processes, at least one must
be correct, and this process locks $x$ before sending its message.
The required then follows from Lemma~\ref{lem:locks-agree}.
\end{proof}

\begin{lemma} \label{lem:equiv}
  For any execution of Fast TetraBFT under partial synchrony, there is an execution of the original TetraBFT  under asynchrony with the exact same set of decisions in views $>0$. Furthermore, if in the execution of Fast TetraBFT all processes have $\gvar{lock} = \bot$ from some point onward, then the execution of TetraBFT is also valid under partial synchrony.
\end{lemma}
\begin{proof}
  Consider any execution of Fast TetraBFT under partial synchrony. We perform the following transformations.
  \begin{itemize}
    \item Remove every action relating to view $0$.
    \item Consider any message $m = \mess{propose}(v, x)$ received by a correct process $p$ with $\gvar{lock} \neq \bot$ and $\gvar{lock} \neq x$, i.e., \autoref{line:safe_follower_lock_check} blocks $p$ from accepting the proposal.
    If $m$ arrives when $p$ is in view $v$ and will unlock during $v$ at a time $t'$, then delay the delivery of $m$ until $t'$.
    If not, then drop $m$.
  \item Consider any time $t''$ when a leader $q$ of view $v'$ could have sent a proposal but does not do so because of its \gvar{lock}.
    That is, no value is considered safe by \func{safe\_val\_leader}, but there is a value $x$ that would be safe if we removed \autoref{line:safe_leader_lock_check}.
    In this case, introduce a $\mess{propose}(v', x)$ message sent by $q$ at $t''$ that is immediately dropped before being received by anyone.
  \end{itemize}
  These transformations yield a valid asynchronous execution of TetraBFT where some correct processes may have different inputs (due to \autoref{line:lock}), but still make the same decisions.
  Assume now that in the execution of Fast TetraBFT all processes have $\gvar{lock} = \bot$ from some time $t$.
  The above transformations do not delay or drop messages sent after $t$.
  Hence, the execution of TetraBFT we constructed is also valid under partial synchrony.
\end{proof}

At the beginning of the proof of {\sf Agreement}, we use the fact that there are no contradictory decisions in views $>0$. This is shown as follows:
\begin{lemma}
    Assume that, in Fast TetraBFT, correct processes $p$ and $q$ decide values $x$ and $y$ in views $> 0$. Then $x = y$.
\end{lemma}
\begin{proof}
  Consider any execution of Fast TetraBFT where $p$ and $q$ decide values $x$ and $y$ in views $> 0$. By \autoref{lem:equiv}, we know there is an execution of TetraBFT with the same decisions. Since TetraBFT is safe in this execution, we must have $x = y$.
\end{proof}

At the beginning of the proof of \autoref{thm:liveness}, we use the fact that, if all correct processes eventually unlock, then they will all decide. This is shown as follows:
\begin{lemma} \label{prop:live_tetra}
  Assume that, in Fast TetraBFT, all correct processes have $\gvar{lock} = \bot$ from some point
  onwards. Then all of them eventually decide.
\end{lemma}
\begin{proof}
  Consider any execution of Fast TetraBFT where the hypothesis holds. By \autoref{lem:equiv}, there is an execution of TetraBFT under partial synchrony with the same set of decisions in views $>0$. Since TetraBFT is live, every correct process decides.
\end{proof}

\fi

\end{document}